\documentclass[reqno]{article}
\usepackage{enumerate}
\usepackage{amsthm}
\usepackage[]{amsmath}
\usepackage{amssymb}
\usepackage{graphicx}
\usepackage{cite}
\newtheorem{fact}{Fact}

\newcommand{\C}{\mathcal{C}}

\newcommand{\G}{\mathbf{G}}
\newcommand{\A}{\mathbf{A}}
\newcommand{\F}{\mathbb{F}}
\newcommand{\RS}{\mathsf{RS}}

\renewcommand{\deg}[1]{\mathsf{deg}\left(#1\right)}

\newtheorem{thm}{Theorem}

\newtheorem{clry}{Corollary}
\newtheorem{lem}{Lemma} 
\begin{document}
\title{Balanced Reed--Solomon Codes}
\author{
Wael Halbawi\thanks{Department of Electrical Engineering - California Institute of Technology - Email: \texttt{\{whalbawi, hassibi\}@caltech.edu}} , 
Zihan Liu\thanks{Department of Information Engineering - The Chinese University of Hong Kong - Email: \texttt{1155029096@link.cuhk.edu.hk}} ,
Babak Hassibi\footnotemark[1]}
\maketitle

\begin{abstract}
We consider the problem of constructing linear Maximum Distance Separable (MDS) error-correcting codes with generator matrices that are sparsest and balanced. In this context, sparsest means that every row has the least possible number of non-zero entries, and balanced means that every column contains the same number of non-zero entries. Codes with this structure minimize the maximal computation time of computing any code symbol, a property that is appealing to systems where computational load-balancing is critical. The problem was studied before by Dau et al. where it was shown that there always exists an MDS code over a sufficiently large field such that its generator matrix is both sparsest and balanced. However, the construction is not explicit and more importantly, the resulting MDS codes do not lend themselves to efficient error correction. With an eye towards explicit constructions with efficient decoding, we show in this paper that the generator matrix of a cyclic Reed--Solomon code of length $n$ and dimension $k$ can always be transformed to one that is both sparsest and balanced, for all parameters $n$ and $k$ where $\frac{k}{n}(n - k + 1)$ is an integer.
\end{abstract}

\section{Introduction}
We study the problem of constructing a linear error-correcting code in which each code symbol is a function of exactly the same number of message symbols and each message symbol is encoded by the least number of code symbols, subject to the constraint that the code is Maximum Distance Separable (MDS). Furthermore, we require the code to be decodable in polynomial time. This setup is natural for applications such as distributed storage networks and wireless sensor networks. As an example, consider a distributed storage network of $n$ nodes deployed to robustly store a data file of size $k$ symbols in a way such that it can be recovered from any $k$ nodes. It is desirable to build a system in which updating a particular message symbol results in updating the data stored on a minimal number of nodes. In addition, if every storage node encodes the same number of symbols, then the maximum time required to compute the symbol stored at any node is minimized. Thus, the computation load is  balanced across the network in the sense that there are no nodes that behave as bottlenecks. Another example is that of a wireless sensor network as presented in~\cite{Dau2013}. A network of $n$ sensors is jointly measuring $k$ parameters, where each sensor measures the same number of parameters and no parameter is measured more than required. The measurements are transmitted to a fusion center for processing, with the requirement being maximal protection against errors.

One would like to give a deterministic construction of an error-correcting code over a small field that behaves in the aforementioned way, and can be decoded efficiently. In this paper, we show how to construct cyclic Reed--Solomon codes in which each code symbol is a function of $\frac{k}{n}(n - k + 1)$ message symbols (balanced), and each message symbol is encoded by $n - k + 1$ code symbols (sparsest). Being Reed--Solomon codes, they can be decoded efficiently and the required field size of the construction scales linearly in the blocklength.

\subsection{Prior Work}
The problem described earlier was first studied by Dau et al. in~\cite{Dau2013}. In that paper, the authors introduced the concept of balanced and sparsest generator matrices for MDS error-correcting codes. A generator matrix of a length $n$ and dimension $k$ MDS code is called \textit{sparsest} if each row contains exactly $n - k + 1$ nonzero entries. In addition, a generator matrix is called balanced if each column contains $\left\lceil \frac{k}{n}(n - k + 1) \right\rceil$ or $\left\lfloor \frac{k}{n} (n - k + 1) \right\rfloor$ nonzero entries. In the same paper, the authors propose an algorithm that produces a sparsest and balanced matrix $\A \in  \left\lbrace0,1\right\rbrace^{k \times n}$. Using a probabilistic argument, they show (non-constructively) that there exists a choice of coefficients from a finite field that can replace the non-zero entries of $\A$ to produce a generator matrix of a code that is MDS. The field size $q$ should satisfy $q > \binom{n-1}{k-1}$. While such codes can be decoded in polynomial time when considered as erasure codes, very little can be said about the complexity needed to decode them from errors.

The problem of decentralized erasure codes was considered in~\cite{Dimakis2006}. In that paper, the authors propose a randomized scheme to construct a linear MDS code that possesses an optimally sparse generator matrix $\G$. In particular, they show that $O(\log k)$ randomly chosen non-zero entries per row of $\G$ is necessary and sufficient to probabilistically guarantee that the code is MDS. The analysis presented shows that this probability increases polynomially in $q$.

A related problem studied in~\cite{Blaum1999,Louidor2006,Blaum1995} is that of finding lowest-density codes that are MDS over $\F_q^b$, an extension of $\F_q$, but are linear over $\F_q$. Motivated by disk arrays, the codes are systematic over $\F_q^b$ and have sparse generator and parity check matrices over $\F_q$. In particular, the authors of ~\cite{Blaum1999} show that for a code of length $n$ over $\F_q^b$, and dimension $kb$ over $\F_q$, at least $k+1$ and $(n - k + 1)$ non-zero entries in its parity-check matrix and generator matrix (both over $\F_q$), respectively are necessary. Furthermore, they present constructions that attain these bounds for a wide range of $n$ and $k$.

Lastly, the problem of constructing error-correcting codes with general support constraints on their generator matrices was studied in~\cite{Dau2014ISIT, Dau2015JSAC} and~\cite{Halbawi2014DRS, Halbawi2014DGC, Halbawi2015ISIT}. This line of work addresses the problem of finding codes of maximum distance possible that adhere to prescribed encoding constraints defined on a bipartite graph. In particular, a minimum distance bound was derived in~\cite{Halbawi2015ISIT} for any systematic code adhering to the given constraints and was shown to be achievable using Reed--Solomon codes.

\section{Problem Setup}
Formally speaking, we would like to construct a linear MDS error-correcting code of length $n$ and dimension $k$ with generator matrix $\mathbf{G} \in \F_q^{k \times n}$ that satisfies the following properties:
\begin{enumerate}[(P1)]
\item Each column of $\G$ has exactly $\frac{k}{n} (n - k + 1)$ non-zero entries. \label{item:balanced}
\item Each row of $\G$ has exactly $n - k  + 1$ non-zero entries. \label{item:sparsest}
\item The code can be decoded in polynomial time and is constructed over $\F_q$, where $q$ scales linearly in $n$.
\end{enumerate}

Sticking to the terminology of \cite{Dau2013}, we call a code \textit{balanced} if it has a generator matrix that  satisfies (P1) and \textit{sparsest} if it satisfies (P2). A natural question to ask is whether previously known MDS codes possess generator matrices that satify properties (P\ref{item:balanced}) and (P\ref{item:sparsest}). In particular, Reed--Solomon codes are an appealing candidate as they are well studied and can be decoded efficiently using a plethora of decoding algorithms~\cite{welch1986, Massey1969, Guruswami1999}.

We will show that for every $n = q - 1$, where $q$ is a prime power, and $k$ is such that $\frac{k}{n} (n-k+1)$ is an integer, the Reed--Solomon code with defining set $\left\lbrace1, \alpha, \ldots, \alpha^{n-1}\right\rbrace$ has a generator matrix that satisfies (P\ref{item:balanced}) and (P\ref{item:sparsest}). When in this form, we will call the code a balanced\footnote{When the code of interest is required to be MDS, a balanced generator matrix as defined by Property (P\ref{item:balanced}) will necessarily be sparsest, eliminating the need to explicitly mention that it possesses this characteristic.} Reed-Solomon code. We begin by formally defining Reed--Solomon codes and then present a few technical results that culminate in the main theorem of this paper.

\section{Preliminaries}

\subsection{Reed--Solomon Codes}
\label{sec:RS}
Throughout this paper, we will refer to a Reed--Solomon code of length $n$ and dimension $k$ over $\F_q$ as $\RS[n,k]_q$. When $n = q - 1$, we simply drop the subscript $q$ and refer to the code as $\RS[n,k]$. We will use the definition of Reed--Solomon codes as in~\cite{Reed1960}. More precisely, $\RS[n,k]_q$ is the the $k$-dimensional subspace of $\mathbb{F}_q^n$ given by $\C_\text{RS} =\left\lbrace (m(\alpha_1), \ldots, m(\alpha_n)):\deg{m(x)} < k \right\rbrace$, where the $m(x)$ are polynomials over $\F_q$ of degree $\deg{m(x)}$, and the $\alpha_i \in \F_q$ are distinct (fixed) field elements. Each message vector $\mathbf{m} = (m_0, \ldots, m_{k-1})$ is mapped to a message polynomial $m(x)=\sum_{i=0}^{k-1} m_i x^i$, which is then evaluated at the $n$ elements $\left\lbrace \alpha_1,\alpha_2,\ldots,\alpha_n \right\rbrace$ of $\mathbb{F}_q$, known as the defining set of the code. The codeword associated with $m(x)$ is $\mathbf{c} = (m(\alpha_1), \ldots, m(\alpha_n))$,  which we also call the evaluation of $m(x)$ at $\left\lbrace \alpha_1,\alpha_2,\ldots,\alpha_n \right\rbrace$. Reed--Solomon codes are MDS codes; their minimum distance attains the Singleton bound, i.e., ${d(\RS[n,k]_q) = n - k + 1}$. For brevity, we set $d = n - k + 1$.

Throughout this paper, the defining set of $\RS[n,k]$ will be chosen to be the set of consecutive powers of a primitive element $\alpha\in \F_q$, i.e., $\left\lbrace1,\alpha,\ldots, \alpha^{n-1}\right\rbrace$. This choice of evaluation points gives rise to a cyclic Reed--Solomon code. The generator matrix of such a code is
\begin{equation}
\label{eqn:generator}
\G_{\text{RS}} = 
\begin{bmatrix}
1 & 1  & \cdots & 1\\
1 & \alpha  & \cdots & \alpha^{n-1}\\
\vdots & \vdots  & \ddots & \vdots\\
1 & \alpha^{(k-1)}   & \cdots & \alpha^{(n-1)(k-1)}
\end{bmatrix}.
\end{equation}

The polynomial nature of Reed--Solomon codes allows us to precisely characterize codewords with a prescribed set of coordinates required to be equal to 0. Formally, suppose we would like to find a vector $\mathbf{c} \in \RS[n,k]$ for which ${c_{j_1} = \cdots= c_{j_l} = 0}$. We let $t(x) = {\prod_{j = 1}^{l}(x - \alpha^{i_j})} = {\sum_{i = 0}^{l} t_ix^i}$, and form the vector of coefficients of $t(x)$ as $\mathbf{t} = (t_0, t_1, \ldots, t_{l})$. The codeword resulting from encoding of $\mathbf{t}$ using $\G_\text{RS}$ is a codeword $\mathbf{c}$ with zeros in the desired coordinates. Indeed, $\mathbf{t} \G_\text{RS}$ is the evaluation of the polynomial $t(x)$ at $\{1, \alpha, \ldots, \alpha^{n-1}\}$. Since $t(x)$ has $\{\alpha^{j_1}, \ldots, \alpha^{j_l}\}$ as roots, it follows that ${[\mathbf{t} \G_\text{RS}]_{j_1} = \cdots = [\mathbf{t} \G_\text{RS}]_{j_l} = 0}$.
 
We present now the BCH bound, a fact that usually accompanies this interpretation of Reed--Solomon codes and which we'll rely on heavily in the construction presented in this paper. The BCH bound gives a lower bound on the number of nonzero coefficients of a polynomial whose roots are of a particular form.

\begin{fact}[BCH Bound]
Let $p(x)$ be a non-zero polynomial (not divisible by $x^{q-1} - 1$) with coefficients in $\F_q$. Suppose $p(x)$ has $t$ (cyclically) consecutive roots, i.e. $p(\alpha^{j+1}) = \cdots = p(\alpha^{j+t})=0$, where $\alpha$ is primitive in $\F_q$. Then at least $t+1$ coefficients of $p(x)$ are non-zero.
\end{fact}

For a proof of the BCH bound, see e.g.,~\cite[p.238]{McEliece}. We will use the BCH bound to show that for a polynomial $p(x)$ whose roots are $k-1$ consecutive powers of $\alpha$, its $k$  scaled versions $\left\lbrace p(\alpha^{j_l}x)\right\rbrace_{l = 1}^k$ are linearly independent. 

\begin{lem}
\label{lem:BCH}
Let $\alpha \in \F_q$ be a primitive element. Let $p(x) = \prod_{i = 0}^{k-2}(x - \alpha^i)$ and define the scaled polynomial $p^{(j_l)}(x) = p(\alpha^{j_l}x)$, and $k < q$. Then, the polynomials $\left\lbrace p^{(j_l)}(x)\right\rbrace_{l = 1}^{k}$ are linearly independent whenever the $j_l\textrm{'s}$ are distinct modulo ${q-1}$.
\end{lem}

\begin{proof}
Let $p(x) = \sum_{i = 0}^{k-1}p_i x^i$. Then, we have $p^{(j_l)}(x)  =\sum_{i = 0}^{k-1}p_i \alpha^{j_li}x^i$.   Form a matrix $\mathbf{P}$ where the $l^\textrm{th}$ row is ${\mathbf{p}_{j_l} = \left(p_0, p_1\alpha^{j_l}, \ldots, p_{k-1}\alpha^{j_l(k-1)}\right)}$, i.e.,
\begin{equation*}
\mathbf{P} = \begin{bmatrix}
p_0 & p_1\alpha^{j_1} & \cdots & p_{k-1}\alpha^{j_1(k-1)}\\
\vdots & \vdots & \ddots & \vdots\\
p_0 & p_1\alpha^{j_k} & \cdots & p_{k-1}\alpha^{j_k(k-1)}\\
\end{bmatrix}.
\end{equation*}

We can write the determinant of $\mathbf{P}$ as 
\begin{equation*}
\mathsf{det}(\mathbf{P}) = \begin{vmatrix}
1& \alpha^{j_1} & \cdots & \alpha^{j_1(k-1)}\\
\vdots & \vdots & \ddots & \vdots\\
1& \alpha^{j_k} & \cdots & \alpha^{j_k(k-1)}\\
\end{vmatrix}
\prod_{i=0}^{k-1}p_i.
\end{equation*}

The matrix in the expression is a Vandermonde matrix and has a nonzero determinant if and only if $\left\lbrace \alpha^{j_1}, \ldots, \alpha^{j_k}\right\rbrace$ are distinct in $\F_q$. Indeed, this is the case when $\alpha$ is a primitive root in $\F_q$ and the exponents are distint modulo $q-1$. Furthermore, the BCH bound in Fact 1 guarantees that the $p_i\text{'s}$ are all nonzero. Therefore, $\mathbf{P}$ is a full rank matrix and the polynomials $\left\lbrace p^{(j_l)}(x)\right\rbrace_{l = 1}^{k}$ are linearly independent over $\F_q$.
\end{proof}

Equipped with this Lemma \ref{lem:BCH}, we deduce that the set of codewords in $\RS[n,k]$ that correspond to message polynomials $\left\lbrace p^{(j_l)}(x)\right\rbrace_{l = 1}^{k}$ are linearly independent. The following corollary is immediate.

\begin{clry}
\label{clry:BCH}
Let $\alpha \in \F_q$ be a primitive element. Let $p(x)$ and $\left\lbrace p^{(j_l)}(x)\right\rbrace_{l = 1}^{k}$ be as in Lemma \ref{lem:BCH}, with $k < q$ and $n = q - 1$. Let $\mathbf{c}_l$ be the evaluation of $p^{(j_l)}(x)$ at $\left\lbrace1, \alpha, \ldots, \alpha^{n-1}\right\rbrace$. Then, the set of codewords $\left\lbrace\mathbf{c}_l\right\rbrace_{l=1}^k$ spans $\RS[n,k]$.
\end{clry}

\section{Construction}
We are now ready to present the main result of the paper. We will constructively prove the following theorem.

\begin{thm}
\label{thm:main}
Let $n = q - 1$ where $q$ is a prime power, and $k$ is such that $\frac{k}{n}(n-k+1)$ is an integer. Then, there exists a Reed--Solomon code $\RS[n,k]_q$ with a generator matrix $\G$ in which every row is has weight $n- k + 1$ and every column has weight $\frac{k}{n}(n-k+1)$.
\end{thm}

We will demonstrate the proof of this theorem after presenting a lemma that is key to the construction. The underlying idea of the construction is to select a set of $k$ vectors from $\{0,1\}^n$, which we call codeword \textit{masks} $\{\mathbf{a}_{j_1}, \ldots, \mathbf{a}_{j_k}\}$, such that when stacked as rows of a matrix $\A_{n,k}$, properties (P\ref{item:balanced}) and (P\ref{item:sparsest}) are satisfied. For each mask $\mathbf{a}_{j_l}$, we select a codeword from $\RS[n,k]$ with zeros in locations as prescribed by the support of $\mathbf{a}_{j_l}$. As mentioned earlier, the polynomial nature of Reed--Solomon codes allows us to accomplish this easily. Once the codewords are fixed, one has to ensure that they span $\RS[n,k]$, resulting in a generator matrix for $\RS[n,k]$ which is both balanced and sparsest.

We start by describing the set of masks that will be used in our construction. Property (P\ref{item:sparsest}) requires that each mask $\mathbf{a}_i$ has $k-1$ zeros. For the case when $\frac{k}{n}(n-k+1)$ is an integer, it turns out that restricting the codeword \textit{masks} to those with $k-1$ cyclically consecutive zeros suffices to construct a balanced matrix $\A_{n,k}$. Henceforth, let $\mathbf{a}$ be a vector of length $n$ with $k-1$ (cyclically) consecutive zeros and $d = n-k+1$ (cyclically) consecutive ones. Let ${\A \in \left\lbrace0,1\right\rbrace^{n \times n}}$ be the circulant matrix, denoted by $\mathsf{circ}(\mathbf{a})$, whose rows $\left\lbrace\mathbf{a}_i\right\rbrace_{i = 0}^{n-1}$ are left circular shifts of $\mathbf{a}$. In particular $\mathbf{a}_i$ is the vector $\mathbf{a}$ shifted $i$ times to the left. Indexing the columns of $\A$ from $0$ to $n-1$, we choose $\mathbf{a}$ in a way so that the $i\text{'th}$ row $\mathbf{a}_i$ has zeros in locations ${\{d - i,\ldots, d - i + k - 2\}}$ modulo $n$.  For example, let $n = 6$ and $k = 3$, then we obtain
\begin{equation*}
\A = 
\begin{bmatrix}
1 & 1 & 1 & 1 & 0 & 0\\
1 & 1 & 1 & 0 & 0 & 1\\
1 & 1 & 0 & 0 & 1 & 1\\
1 & 0 & 0 & 1 & 1 & 1\\
0 & 0 & 1 & 1 & 1 & 1\\
0 & 1 & 1 & 1 & 1 & 0
\end{bmatrix}.
\end{equation*}

Each row of $\A$ corresponds to the mask of a potential codeword of the code's generator matrix as desired. However, we need to jointly select a set of $k$-subset of $\left\lbrace\mathbf{a}_i\right\rbrace_{i = 1}^n$ so that each column of $\G$ has weight $\frac{k}{n}(n-k+1)$. We will now pose a linear system whose solution provides a recipe for choosing the rows of $\G$. Let $b = \frac{k}{n}(n- k + 1)$ and let $\mathbf{b}$ be the all $b$ vector of length $n$. We aim to find a row vector $\mathbf{v} \in \left\lbrace0,1\right\rbrace^n$ of weight exactly $k$ such that the following holds,
\begin{equation*}
\mathbf{v}\A = \mathbf{b}.
\end{equation*}

A solution $\mathbf{v}$ dictates which codewords of $\textsf{RS}[n,k]$ will form $\G$. If $v_i$ is non-zero, then a codeword with mask $\mathbf{a}_i$ is chosen as a row of $\G$. Since $\mathbf{v}$ is $k$-sparse, a set of exactly $k$ codewords is chosen. In our example, $\mathbf{v} = (1,0,1,0,1,0)$ is one such vector, which selects the rows $\left\lbrace\mathbf{a}_0,\mathbf{a}_2,\mathbf{a}_4\right\rbrace$. The corresponding matrix is,
\begin{equation*}
\mathbf{A}_{6,3} = 
\begin{bmatrix}
1 & 1 & 1 & 1 & 0 & 0\\
1 & 1 & 0 & 0 & 1 & 1\\
0 & 0 & 1 & 1 & 1 & 1
\end{bmatrix}.
\end{equation*}

This particular choice of $\mathbf{v}$ generalizes to arbitrary $n$ and $k$ when $\frac{k}{n}(n - k +1)$ is an integer. We present this fact formally in the following lemma.

\begin{lem}
\label{lem:balanced}
Let $n,k$ be such that $b = \frac{k}{n}(n - k + 1) \in \mathbb{Z}$, with $\mathbf{b}$ being the all $b$ vector. Let $\mathbf{a} = (1,\ldots,1,0,\ldots,0)$ be $d$-sparse and $\mathbf{\A} = \mathsf{circ}(\mathbf{a})$. Then a solution to $\mathbf{v}\A = \mathbf{b}$ exists with $\mathbf{v} \in \left\lbrace0,1\right\rbrace^n$ being $k$-sparse. Furthermore, one such solution is $\mathbf{v} = (\mathbf{1},\mathbf{0},\mathbf{1},\mathbf{0},\ldots,\mathbf{1},\mathbf{0})$ where $\mathbf{1}$ is the all-one vector of length $\frac{k}{g}$ and $\mathbf{0}$ is the all-zero vector of length $\frac{d-1}{g}$, and $g = \mathsf{gcd}(k,n)$.
\end{lem}

\begin{proof}
Let $k = \beta_k g$, $d - 1 = \beta_d g$, and $\beta_n = \beta_k + \beta_d$. Since $b = \frac{k}{n}d =  \frac{\beta_k}{\beta_d + \beta_k}d \in \mathbb{Z}$ and ${\beta_d + \beta_k} \nmid {\beta_k}$, we have that $\beta_n \mid d$, and so $d = \beta_n m$ for some $m$. Since $\A$ is a circulant matrix, we have that $\A^\mathsf{t} = \A$ and so the first column of $\A$ is precisely $\mathbf{a}^\mathsf{t}$. Fix $\mathbf{v}$ as in the statement of the lemma and consider the product between $\mathbf{v}$ and $\mathbf{a}^\mathsf{t}$,
\begin{align}
\mathbf{v} \cdot \mathbf{a}^\mathsf{t} &= \sum_{i = 0}^{n - 1}v_ia_i \nonumber\\ 
 &= \sum_{i = 0}^{d - 1} v_i \label{eqn:d-sparse} \\ 
 &= \sum_{j = 1}^m \sum_{i = 1}^{\beta_k}1  \label{eqn:countability} \\
 &= \beta_k m  \nonumber  \\
 &= \beta_k\frac{d}{\beta_n} = b. \nonumber 
\end{align}

where \eqref{eqn:d-sparse} follows from the fact that only the first $d$ entries of  $\mathbf{a}^\mathsf{t}$  are non-zero. We have that $\mathbf{v}$ is composed of an alternating sequence of $\mathbf{1}\text{'s}$ and $\mathbf{0}\text{'s}$ with lengths $\beta_k$ and $\beta_d$, respectively. Since $\beta_n | d$, we deduce that the first $d$ entries of $\mathbf{v}$ are precisely $m$ copies of $(\mathbf{1},\mathbf{0})$. As a result, we obtain equation \eqref{eqn:countability}. Thus, we have that the first entry of the $\mathbf{v}\A$ is indeed equal to $b$.

Now fix an arbitrary column $\tilde{\mathbf{a}}^\mathsf{t}$ of the matrix $\A$ where the sequence of ones starts at position $l$ and ends at $((l + d - 1))_n$, and $((\cdot))_n$ refers to reducing the argument modulo $n$. The product between $\mathbf{v}$ and $\hat{\mathbf{a}}^\mathsf{t}$ results in\footnote{We did not reduce the indices of the sum in \eqref{lem:eqn:mod n} modulo $n$ to provide a more concise proof. For any index $i \geq n$, define $v_i = v_{((i))_n}$.},
\begin{align}
\mathbf{v} \cdot \tilde{\mathbf{a}}^\mathsf{t} &= \sum_{i = l}^{l + d - 1} v_i   \label{lem:eqn:mod n} \\
&= \sum_{i \in \{l, \ldots, l + d - 1\}_d} v_i  \label{lem:eqn:mod d}\\
&= \sum_{i =0}^{d-1} v_i = b.  \nonumber 
\end{align}

Since $\mathbf{v}$ is periodic with period $\beta_n$, and as $\beta_n \mid d$, it is also periodic with period $d$. Therefore, $v_{i} = v_{((i))_d}$ for ${i \in \{l, \ldots, l + d - 1\}}$ which justifies \eqref{lem:eqn:mod d}.  Lastly, since $\left\lbrace l,\ldots, l + d - 1\right\rbrace$ are distinct modulo $d$, i.e. is a complete residue system modulo $d$, we can reindex the sum to run over $\{0,\ldots, d-1\}$, and we know that it is equal to $b$. Since $\tilde{\mathbf{a}}^\mathsf{t}$ was an arbitrary column of $\A$, we have established that $\mathbf{v}\A = \mathbf{b}$.
\end{proof}

We are now ready to prove the main theorem of the paper.

\begin{proof}[Proof of Theorem 1]
Fix $k$ and $n = q - 1$ such that ${b = \frac{k}{n}(n - k + 1)}$ is an integer. Construct the matrix $\A$ as in Lemma \ref{lem:balanced}, which guarantees the existence of a $k$-sparse solution in $\{0,1\}^n$ to $\mathbf{v}\A = \mathbf{b}$. As mentioned earlier, the solution $\mathbf{v}$ is used to construct a balanced matrix $\A_{n,k}$, and indicates which codewords from $\RS[n,k]$ are to be selected to form the generator matrix $\G$ which satisfies (P\ref{item:balanced}) and (P\ref{item:sparsest}). Let the support of $\mathbf{v}$ be $\{j_1,\ldots,j_k\}$, which implies that codewords with masks ${\{\mathbf{a}_{j_1},\ldots,\mathbf{a}_{j_k}\}}$ are to be chosen. By construction $\mathbf{a}_{j_l}$ has zeros in locations ${\{d - j_l, \ldots, d - j_l + k - 2\}_n}$. Let ${p(x) = \prod_{i = 0}^{k-2}(x - \alpha^i)}$. For each mask $\mathbf{a}_{j_l}$, form the polynomial $p^{(j_l)}(x) = p\left(\alpha^{-(d - j_l)}x\right)$. We then have,
\begin{align}
p^{(j_l)}(x) &= p\left(\alpha^{-(d - j_l )}x \right) \nonumber \\
&= \prod_{i = 0}^{k - 2} \left(\alpha^{-(d - j_l)}x - \alpha^i \right) \nonumber \\
&= C_{j_l} \prod_{i = 0}^{k - 2} \left(x - \alpha^{(d - j_l + i)} \right) \nonumber \\
&= C_{j_l} \prod_{i = d - j_l }^{d - j_l + k - 2} \left(x - \alpha^{i} \right) \label{eqn:scaled_roots}\\
&= \sum_{i = 0}^{k-1} p_{l,i} x^i.
\end{align}

where $C_{j_l}$ is a non-zero constant. Indeed, we have that $p^{(j_l)}(x)$ vanishes at $\left\lbrace\alpha^{d - j_l},\ldots,\alpha^{d - j_l +k - 2}\right\rbrace$, as is evident from \eqref{eqn:scaled_roots}. Now let ${\mathbf{p}_{l} = (p_{l,0},\ldots,p_{l,{k-1}})}$ be the $l^{\text{th}}$ row of the matrix $\mathbf{P}$. We let the generator matrix $\mathbf{G}$ be $\mathbf{P}\mathbf{G}_\text{RS}$, where $\mathbf{G}_\text{RS}$ is defined as in \eqref{eqn:generator}. As described earlier, we have that the $(l,j)^\text{th}$ entry of $\mathbf{G}$ is equal to 0 if and only if $p^{(j_l)}(x)$ vanishes at $\alpha^{j}$. It now follows that $\mathbf{G}$ is sparsest and balanced as desired since it has the mask of $\A_{n,k}$. Indeed, the matrix  $\mathbf{P}$ is invertible by Lemma \ref{lem:BCH} and so $\G$ is full rank by Corollary \ref{clry:BCH}. As a result, the dimension of the code generated by $\G$ is $k$ and so it spans $\RS[n,k]$.
\end{proof}

\section{Example}
In this section, we will construct a balanced and sparsest generator matrix for $\RS[6,3]$. For these parameters, we have ${b = \frac{k}{n}(n - k + 1) = 2}$. The matrix of potential codeword masks is
\begin{equation*}
\A = 
\begin{bmatrix}
1 & 1 & 1 & 0 & 0 & 0\\
1 & 1 & 0 & 0 & 0 & 1\\
1 & 0 & 0 & 0 & 1 & 1\\
0 & 0 & 0 & 1 & 1 & 1\\
0 & 0 & 1 & 1 & 1 & 0\\
0 & 1 & 1 & 1 & 0 & 0
\end{bmatrix}.
\end{equation*}

Since $d = 3$, $g = 3$, and $\beta_k = \beta_d = 1$, a $4-$sparse  solution in $\{0,1\}^6$ to ${\mathbf{v}\A =\mathbf{b}}$ is ${\mathbf{v} = (1,1,0,1,1,0)}$. The resulting set of masks is ${\{\mathbf{a}_0, \mathbf{a}_1, \mathbf{a}_3, \mathbf{a}_4\}}$ which in matrix form is given by
\begin{equation}
\A_{6,4} = 
\begin{bmatrix}
1 & 1 & 1 & 0 & 0 & 0\\
1 & 1 & 0 & 0 & 0 & 1\\
0 & 0 & 0 & 1 & 1 & 1\\
0 & 0 & 1 & 1 & 1 & 0\\
\end{bmatrix}.
\end{equation}

Now fix $\alpha \in \F_7$ to be a primitive element. For example, ${\alpha = 3}$. The codewords to be chosen are associated with the polynomials 
\begin{align*}
p^{(0)}(x) &= 6(x - 6)(x - 4)(x - 5),\\
p^{(1)}(x) &= (x - 2)(x - 6)(x - 4),\\
p^{(3)}(x) &= (x - 1)(x - 3)(x - 2),\\
p^{(4)}(x) &= 6(x - 5)(x - 1)(x - 3).
\end{align*}

The transformation matrix $\mathbf{P}$ resulting from this set of polynomials is given by 
\begin{equation*}
\mathbf{P} = 
\begin{bmatrix}
1 & 3 & 1 & 6 \\
1 & 2 & 2 & 1 \\
1 & 4 & 1 & 1 \\
1 & 5 & 2 & 6
\end{bmatrix}.
\end{equation*}

Finally, the balanced and sparsest generator matrix obtained using $\mathbf{P}$ is
\begin{equation*}
\mathbf{G} = 
\begin{bmatrix}
4 & 6 & 3 & 0 & 0 & 0 \\
6 & 3 & 0 & 0 & 0 & 4 \\
0 & 0 & 0 & 4 & 6 & 3 \\
0 & 0 & 4 & 6 & 3 & 0
\end{bmatrix}.
\end{equation*}

\section{Discussion}
As mentioned, the construction presented assumes that $\frac{k}{n}(n - k + 1)$ is an integer, which can be restrictive in certain practical scenarios. Nonetheless, one can think of a greedy algorithm based on heuristics that produces a balanced and sparsest generator matrix for any cyclic Reed--Solomon code. We are currently investigating the correctness of one particular algorithm and the results will be reported in a subsequent paper. In addition, it is desirable to relax the cyclic condition on the underlying $\RS[n,k]_q$ to provide felixibility in choosing its defining set. Note that the technique presented in Lemma~\ref{lem:balanced} does not assume any constraint on $n$. As a result, it can be used to find a balanced matrix $\A_{n,k}$ for any $n$ and $k$ as long as $\frac{k}{n}(n - k + 1)$ is an integer. The difficulty becomes evident when trying to derive an analogue of Lemma~\ref{lem:BCH} for a polynomial with an arbitrary set of roots. While Lemma~\ref{lem:balanced} does provide a way to construct a balanced Reed--Solomon code with required minimum distance, it does not provide any guarantee on its dimension. Furthermore, it is of particular interest to see whether the techniques presented in this paper extend too other families of known MDS codes. The construction presented here can be generalized to handle Gabidulin codes~\cite{gabidulin1985theory}, leveraging the techniques of~\cite{Halbawi2014DGC}. This family of codes is pertinent to error correction in network coding settings~\cite{Silva2008} in which a balanced code construction can prove useful.
\section{Conclusion}
We have studied the problem of constructing linear MDS error-correcting codes of length $n = q - 1$ and dimension $k$ over small fields that have the sparsest possible generator matrices, and each column has exactly the same number of non-zero entries, when permitted by $k$ and $n$. In particular, we have shown that a Reed--Solomon code of the same parameters always possesses a generator matrix that is both sparsest and balanced. As a result, we do not incur any extra cost in the required field size, and the codes can be decoded using any Reed--Solomon decoder. It remains to show whether the assumption $n = q - 1$ can be lifted to determine whether non-cyclic Reed--Solomon codes, or other known MDS codes, also possess a balanced and sparsest generator matrix.
\bibliographystyle{IEEEtran}
\bibliography{IEEEabrv,library}
\end{document}